\documentclass[9pt,table,xcdraw,compsoc,authorversion,nonacm]{acmart}

\acmConference[preprint]{preprint}

\usepackage[disable]{todonotes}

\usepackage{graphicx}
\usepackage{microtype}
\usepackage{hyperref}
\usepackage{mathtools}
\usepackage{amsmath}
\usepackage{amsthm}
\usepackage[altpo]{backnaur}
\usepackage{listings}

\usepackage{multirow}
\usepackage{tabularx,adjustbox,booktabs}
\usepackage{siunitx}
\usepackage{dcolumn}

\usepackage{eufrak}

\usepackage{array}
\newcolumntype{$}{>{\global\let\currentrowstyle\relax}}
\newcolumntype{^}{>{\currentrowstyle}}

\usepackage[]{xcolor}
\usepackage{colortbl}

\usepackage[justification=centering]{caption}

\usepackage[breakable, skins,xparse]{tcolorbox}

\usepackage{tikz}
\usetikzlibrary{shapes,arrows}

\RequirePackage{titlesec}
\titlespacing*{\section}{2pt}{1.5ex}{1.5ex}
\titlespacing*{\subsection}{2pt}{1ex}{1ex}

\newcolumntype{L}{>{\phantom{$-$}}l}       

\newcommand\restr[2]{{
  \left.\kern-\nulldelimiterspace 
  #1 
  \vphantom{\big|} 
  \right|_{#2} 
  }}
 
\usepackage{natbib}
\usepackage{listings}

\usepackage[linesnumbered,ruled,vlined]{algorithm2e}

\newcommand{\mde}{model-driven engineering}
\newcommand{\mmodel}{metamodel}

\newcommand{\eo}{edit operation}
\newcommand{\Eo}{Edit operation}

\newcommand{\scg}{simple change graph}
\newcommand{\Scg}{Simple Change Graph}

\newcommand{\SCG}{SCG}

\newcommand{\eeo}{edit pattern}

\newcommand{\ceo}{completion operation}
\newcommand{\approach}{our approach}
\newcommand{\edgel}{EdgeList}

\newcommand{\generationPhase}{generation phase}
\newcommand{\trainingPhase}{training phase}

\newcommand{\GenerationPhase}{Generation Phase}
\newcommand{\TrainingPhase}{Training Phase}

\newcommand{\ada}{{\sf text-ada-001}}
\newcommand{\curie}{{\sf text-curie-001}}
\newcommand{\davinci}{{\sf text-davinci-003}}

\newcommand{\numCompletionModels}{two}
\newcommand{\averageCorrectCompletion}{87.60}
\newcommand{\averageRankCompletion}{1.55}
\newcommand{\averageNumberCompletions}{2.71}

\newcommand{\averageCorrectAda}{2.17}

\newcommand{\averageCorrectCurie}{2.13}

\newcommand{\averageCorrectDavinci}{1.00}

\newcommand{\averageTokenAccuracy}{0.969}
\newcommand{\minTokenAccuracy}{0.921}
\newcommand{\maxTokenAccuracy}{0.990}

\newcommand{\avgIncorrectScgs}{0.901}
\newcommand{\avgIncorrectGraphs}{2.26}

\newcommand{\totalCost}{347\$}

\newcommand{\figref}[1]{Figure~\ref{#1}}
\newcommand{\secref}[1]{Section~\ref{#1}}
\newcommand{\defref}[1]{Definition~\ref{#1}}
\newcommand{\tblref}[1]{Table~\ref{#1}}

\newcommand{\countRepo}{24}

\newtheorem{theorem}{Theorem}

\usepackage{ifthen}
\newboolean{extension}
\setboolean{extension}{true}
\newcommand{\inExtension}[2]{\ifthenelse{\boolean{extension}}{\color{black}#1\color{black}}{\color{red}#2\color{black}}}
\definecolor{darkblue}{rgb}{0,0,.75}

\definecolor{eminence}{RGB}{108,48,130}

\lstnewenvironment{PseudoCode}[1][]
{\lstset{language=Matlab,basicstyle=\scriptsize, keywordstyle=\color{darkblue},numbers=left,xleftmargin=.04\textwidth,#1}}
{}

\theoremstyle{definition}
\newtheorem{definition}{Definition}[section]

\begin{document}


%
\title{Towards Automatic Support of Software Model Evolution with Large Language~Models}












\author {Christof Tinnes}
\affiliation {%
\institution {Siemens AG}
\department {Technology}
\city {Munich}
\country {Germany}
}
\affiliation {%
\institution {Saarland University}
\department {Informatics Campus}
\city {Saarbrücken}
\country {Germany}
}
\email {christof.tinnes@siemens.com}

\author {Thomas Fuchß}
\affiliation {%
\institution {HKA Karlsruhe}
\city {Karlsruhe}
\country {Germany}
}
\email {thomas.fuchss@hka-iwi.de}

\author {Uwe Hohenstein}
\affiliation {%
\institution {Siemens AG}
\department {Technology}
\city {Munich}
\country {Germany}
}
\email {uwe.hohenstein@siemens.com}

\author {Sven Apel}
\affiliation {%
\institution {Saarland University}
\department {Informatics Campus}
\city {Saarbrücken}
\country {Germany}
}
\email {apel@cs.uni-saarland.de}
%



\begin{abstract}
Modeling structure and behavior of software systems plays a crucial role, in various areas of software engineering. As with other software engineering artifacts, software models are subject to evolution. Supporting modelers in evolving models by model completion facilities and providing high-level edit operations such as frequently occurring editing patterns is still an open problem. Recently, large language models (i.e., generative neural networks) have garnered significant attention in various research areas, including software engineering. In this paper, we explore the potential of large language models in supporting the evolution of software models in software engineering. We propose an approach that utilizes large language models for model completion and discovering editing patterns in model histories of software systems. Through controlled experiments using simulated model repositories, we conduct an evaluation of the potential of large language models for these two tasks. We have found that large language models are indeed a promising technology for supporting software model evolution, and that it is worth investigating further in the area of software model evolution.
    
\end{abstract}


\maketitle              
\newcommand{\code}[1]{\small {\sl{#1}}}

\section{Introduction}
Models play an important role in modern software and system development~\cite{RodriguesDaSilva2015}, 
software documentation~\cite{kruchten1995, UML2017}, 
system architecture~\cite{SysML2019}, 
simulation~\cite{dabney2004mastering}, 
industry automation~\cite{iec61131plc}, 
and user interface layout design~\cite{Chang2010}. 
For disambiguation, we refer to these models as \emph{software models}.

All artifacts in software and system development are typically subject to evolution, which also applies to software models~\cite{visser2007model}: 
Software models have to evolve because of changing requirements, but they can also be subject to bugfixes or refactorings. For example, a shopfloor automation system has to adapt to changed production processes (i.e., changing requirements), or the handling of an exceptional situation has to be corrected (i.e., a bugfix).

From the perspective of a user of a modelling tool, we can understand the evolution of software models as a sequence of \emph{\eo{s}}:
To change or evolve the model, the tool user executes \eo{s} (e.g., using mouse clicks and keyboard strokes) provided by the tool. For example, the user could modify the velocity of a shopfloor conveyor belt, or add another sensor input to the warning mechanism of a large drive. \Eo{s} can be specified by in-place model transformations~\cite{Mens2006,czarnecki2006feature}, which in turn can be formalized as graph transformations \cite{Ehrig2004,Biermann2012, ehrig2005generation}.

For the evolution of software models, modelling tools typically provide an initial set of \eo{s} (e.g., adding an attribute to an UML class). Nevertheless, since the usage of a (domain-specific) language is also subject to evolution and since (project-specific) usage patterns might emerge, this initial set of \eo{s} can rarely be exhaustive. For example, in object-oriented design, design patterns~\cite{gamma1995design} are widely used but are not part of the UML language specification~\cite{UML2017}. Likewise, in electronic circuit design, patterns such as a low-pass filter will often be used and might not be part of the tool`s standard component library (e.g., consisting of capacitors, resistors, diodes). 
Consequently, approaches for the specification of new user-specific model transformations---and \eo{s}, in particular---have been developed~\cite{arendt2010henshin, Balogh2006, kolovos2012, omg2013qvt}. 

Using a specification language for defining \eo{s} poses two challenges:
First, specifying new \eo{s} requires the knowledge of the specification language and of the concrete domain-specific language (or \mmodel{}) for which the \eo{s} have to be specified. Second, domain-specific \eo{s} are often not explicitly known, that is, they are a form of tacit knowledge~\cite{Polanyi1958}. Externalizing the knowledge can therefore be hard or even impossible for domain experts.

While it might even be necessary to define project-specific \eo{s}~\cite{Tinnes2021,Tinnes2022} the overhead for their manual specification is often not economically feasible for many projects and tool providers~\cite{Kehrer2017, Mokaddem2018, Kappel2012}.
Therefore, researchers have sought to support or even automate the specification of model transformations and \eo{s}, in particular.
These approaches range from a visual or concrete syntax-based support \cite{Avazpour2015, acrectoaie2018vmtl, holldobler2015systematically}, over supervised approaches that turn single examples \cite{Langer2009, Gray2011, Kappel2012} or a set of examples \cite{Kehrer2017, Mokaddem2018} into model transformation specifications (so-called model transformation by-example (MTBE)), to unsupervised approaches such as the generation of model transformations out of a given meta-model \cite{Kehrer2016, mazanek2009generating, kehrer2013generating}.
More recently, mining model transformations from the software model history using frequent subgraph mining has been proposed~\cite{Tinnes2021}.
Mining approaches are especially appealing since they do not require any manual specification, no hand-crafting of examples (as in MTBE), and they are not limited to simple well-formedness rules that can be derived out-of-the \mmodel{}.

Unfortunately, frequent subgraph mining and related techniques are not scalable; in particular large differences between two model revisions might lead to extremely long execution times or large memory consumption~\cite{Tinnes2021}. 
Furthermore, mining approaches lack abstraction capabilities. For example, two \eo{s} that are identical but for one type---even with common parent type in the \mmodel{}---are seen as two different \eo{s} and could even be missed by the approach. 



From the perspective of software model evolution it would also be desirable to have context-dependent \emph{completions}, instead of learning a fixed set of \eo{s}. That is, the user would benefit from the recommendation of a \emph{completion operation} that complements the \eo{s} that have been applied to the model.
    
Motivated by recent breakthroughs of generative neural network architectures for various tasks such as code generation and completion~\cite{chen2021codex}, program repair~\cite{mashhadi2021programrepair}, and even to the generation of graph like structures---most notably molecular generation~\cite{bagal2021molgpt}---we investigate feasibility of generative language models for \eeo{} mining and software model completion. The rationale is that, if we fine-tune the language model to generate completions for software models, it should have learned the underlying patterns of the (modelling) language in question. We propose and study an approach to learn edit and completion patterns from the software model history and to extract the corresponding operations from a generative language model. Specifically, we define and use an encoding for serializations of model difference graphs that we will use to fine-tune the language models. Using a synthetic dataset with control over the actually performed \eo{s}, we then evaluate whether we can generate serializations of \eo{s} using the fine-tuned language models and whether we can complete software model changes correctly. We find that, regarding completion, the approach performs surprisingly well and generates the correct completion in \averageCorrectCompletion{}\% of our samples. The approach is also able to generate \eo{s}, although, in some cases, not all applied \eo{s} were generated.

In a summary, we make the following contributions:
\begin{itemize}
    \item We formalize the concepts of \eo{s}, \ceo{s}, and \eeo{s} such that they are suitable for the mining perspective, and show that this formalization is consistent with earlier constructions. 
    \item We propose an approach to use large language models for mining \eeo{s} and software model completion.
    \item We evaluate the approach in a controlled experiment. We find that the approach can provide correct software model completions and also generate \eo{s} that have been applied in the simulation of the software model repositories.
\end{itemize}

\section{Foundations and State-of-the-Art}\label{sec:background}
\subsection{Software Models and Edit Operations}\label{sec:eos}
In this section we will provide background about software models and we present a new formalization of \eo{s} that is well suited for studying mining approaches. We will then see that this formalization is compatible with earlier formalizations.

In \mde{} the language for a software model (i.e., its abstract syntax and static semantics) is typically defined by a \mmodel{} $\mathcal{TM}$. A model is an instance of its \mmodel{}. 
We denote by $\mathcal{M}$ the set of all valid models (according to some \mmodel{}). 
This can be formalized using typed attributed graphs~\cite{Biermann2012,Ehrig2004}.

\begin{definition}[Abstract Syntax Graph]\label{def:asg}
    An \emph{abstract syntax graph} $G_m$ of a model $m \in \mathcal{M}$ is a typed attributed graph, typed over an attributed type graph $TG$ that is given by the \mmodel{} $\mathcal{TM}$.
\end{definition}

The idea of typed graphs is to define a graph homomorphism (i.e., a function from the typed graph $G$ to the type graph $TG$). Details of this formalization are given in the work by Biermann et al.~\cite{Biermann2012}.
The abstract syntax graph of a model (together with the type graph) contains all information that a model contains. The \emph{concrete (visual) syntax} of a model is only an external representation on the model that provides visual cues for the end-user of a model. It is meant to allow for easier handling and understanding of the models and therefore sometimes even hides information of the model. 

In the present work, we are concerned with model repositories. We assume that the modelling tool already takes care of checking the correct typing of the models, and we therefore expect that the models are correctly typed. We therefore work with a simplified graph representation of the models in which the abstract syntax graph is just a \emph{labeled directed graph} with node labels equal to the node type names and edge labels equal to the edge type names of the abstract syntax graph from \defref{def:asg}. 

\begin{definition}[Labeled Directed Graph]
    Given a label alphabet $L$, a \textit{labeled directed graph} $G$ is a tuple $(V,E,\lambda)$, where $V$ is a finite set of nodes, $E$ is a subset of $V \times V$, called the edge set, and $\lambda: V \cup E \to L$ is the labeling function, which assigns a label to nodes and edges. 
\end{definition}

We refer to the set of all directed labeled graphs by $\mathcal{G}$. In the following, we will use the labeled graphs corresponding to the model instead of the graphs from \defref{def:asg}.
Rather than working directly on the abstract syntax graph of the models, we will mostly be working with model differences.
\begin{definition}[Structural Model Difference]
    A \emph{structural model difference} $\Delta_{mn}$ of a pair of model versions $m$ and $n$ is obtained by matching corresponding model elements in the model graphs $G_{m}$ and $G_{n}$ (using a model matcher~\cite{stephan2013survey}, e.g., EMFCompare~\cite{brun2008model} or SiDiff~\cite{schmidt2008constructing}). Then, there are added elements (the ones present in $G_{n}$ but not in $G_{m}$), removed element (the ones present in $G_{m}$ but not in $G_{n}$), and preserved elements which are present in $G_{m}$ and $G_{n}$.
\end{definition}

We assume here that this matching is deterministic, that is, given two models $m, n \in \mathcal{M}$, we get a unique structural model difference $\Delta_{mn}$.
The structural model difference can be represented as a {\em difference graph}~\citep{ohrndorf2021history} $G_{\Delta{mn}}$, where the nodes carry some prefix label ``Add'', ``Preserve'', or ``Remove'', and matching elements (i.e., the preserved ones) from $G_{m}$ and $G_{n}$ are unified with each other (i.e., the will be present only once). 

We define a \emph{\scg{}} to be the smallest subgraph comprising all changes in the difference graph $G_{\Delta_{mn}}$. 

\begin{definition}[\Scg{}]
Given a difference graph $G_{\Delta_{mn}}$, a \emph{\scg{}} $SCG_{\Delta_{mn}} \subseteq G_{\Delta_{mn}}$ is derived from $G_{\Delta_{mn}}$ by first selecting all the elements in $G_{\Delta_{mn}}$ representing a change (i.e., added, removed nodes and edges) and, second, adding preserved nodes that are adjacent to a changed edge. The \scg{} is the smallest subgraph of $G_{\Delta_{mn}}$ containing all changed nodes and edges.
\end{definition}

\begin{definition}[Endogenous model transformation]
An \emph{endogenous model transformation} is a pair $t = (m,n) \in \mathcal{M} \times \mathcal{M}$. We call $m$ the \emph{source model} and $n$ the \emph{target model} of the transformation and $\mathcal{T} \stackrel{\text{def}}{=}\mathcal{M} \times \mathcal{M}$ the space of endogenous model transformations.
\end{definition}

We can then also define a function $SCG\colon \, \mathcal{T} \to \mathcal{G}$ that takes a model transformation (i.e., a pair of models) as input and returns the simple change graph for the corresponding model difference.

We can use this map $SCG$ to define an equivalence relation on $\mathcal{T}$ by
\begin{alignat*}{2}
&t_1 = (m,n) &&\sim t_2 = (k,l)\,, \qquad \text{if and only if}\, \\
&SCG_{\Delta_{mm}} &&= SCG_{\Delta_{kl}}.
\end{alignat*}
It is straightforward to see that this relation indeed defines an equivalence relation (i.e., the relation is reflexive, symmetric, and transitive). We can therefore define the quotient set $\mathcal{T}/{\sim}$, which---by construction---is set isomorphic to the set of \scg{s}, that is, the range of the map $SCG$.

We can use this construction to formally define the concept of an \emph{\eo{}}.

\begin{definition}
    An \emph{\eo{}} is an equivalence class in the set $\mathcal{E}\stackrel{\text{def}}{=} \mathcal{T}/{\sim}$. An \eo{} is therefore a set of model transformations that have the same simple change graph.
\end{definition}

We can also interpret an \eo{} as a template for a rule to transform a model $m$ into a model~$n$:

For a \scg{}, we call the subgraphs of ``Remove'' and ``Preserve'' nodes the left-hand side graph $L$, and the ``Add'' and ``Preserve'' nodes the right-hand side graph~$R$. The embedding of $L$ and $R$ along the preserved nodes $K$ (i.e., $L \hookleftarrow K \hookrightarrow R$, where $\hookleftarrow$ denotes an injective homomorphism, i.e., an embedding) in the \scg{} defines how to remove the ``Remove'' nodes from $m$ and glue the ''Add'' nodes along $K$. Given an \eo{} $\varepsilon$ and a concrete model $m$, one can define a matching $\operatorname{match}: L~\hookrightarrow~G_{m}$, and perform the removal of ``Remove'' nodes and the gluing of ``Add'' nodes as defined by the \scg{} corresponding to $\varepsilon$, and then set concrete attributes. This yields the corresponding model $n$ with $(m, n) \in \varepsilon$, and this way an \eo{} $\varepsilon \in \mathcal{E}$ can be interpreted as a template for a model transformation in agreement with previous constructions \cite{Biermann2012,Kehrer2015,Tinnes2021}. We therefore also write $m \stackrel{\varepsilon}{\to} n$ to denote a concrete element (i.e., a model transformation) in the equivalence class $\varepsilon \in \mathcal{E}$.

The set of \eo{s} obtained by this construction is huge---infinite to be more precise---and it contains also ``operations'' such as constructing a large model from scratch (i.e., taking $m \coloneqq \bot$, the empty model and $n \in \mathcal{M}$ a large model). This does not coincide with more pragmatic definitions, for example, as the set of operations provided by the modelling tool, which typically is finite. 

A transition from $m \stackrel{\varepsilon}{\to} n$ can usually also be obtained by a sequence (aka. \emph{edit script} \cite{Kehrer2013EditScript}) $$m \stackrel{\varepsilon_{1}}{\to} m_{1} \stackrel{\varepsilon_{2}}{\to} \dots  \stackrel{\varepsilon_{k-1}}{\to}m_{{k-1}}\stackrel{\varepsilon_k}{\to}n.$$


\begin{definition}
    A set $S \subset \mathcal{E}$ is called a \emph{generator} for $\mathcal{E}$, if every model can be reached by a sequence of \eo{s} in $S$, that is,

    \begin{alignat*}{2}
    &\forall \, m \in \mathcal{M}.  &&\exists \,\varepsilon_1,\dots,\varepsilon_{k_m} \in S, \, m_{1},\dots,m_{{k_m-1}} \in \mathcal{M}\\
     &\text{such that}  \quad  &&\bot \stackrel{\varepsilon_1}{\to} m_{1} ‚\stackrel{\varepsilon_2}{\to}  \dots  \stackrel{\varepsilon_{k_m-1}}{\to} m_{{k_m-1}}\stackrel{\varepsilon_k}{\to}m.
    \end{alignat*}
\end{definition}

An example for a generator $S$ is the set of elementary \eo{s} that can be derived from a \mmodel{}, as given in the work of Kehrer et al.\cite{Kehrer2016}. This set is finite and contains rather fine-grained \eo{s}. Any subset $E \subset \mathcal{E}$ can be completed to a generator by joining it with an existing generator (e.g., the set of elementary \eo{s}), therefore ensuring that all valid models can be reached. 

In this work, we explore the continuum between elementary edit operations and the set $\mathcal{E}$ of all edit operations. Our goal is to identify those higher-level \eo{s} whose effect is actually observable in model histories, providing empirical evidence that these are meaningful edit operations from a modeler's point of view. We call these higher-level \eo{s} \emph{\eeo{s}}.
Furthermore, we are interested in completing software models, that is, for an observed evolution $m \stackrel{\varepsilon}{\to} n$, we want to find a completion $\gamma \in \mathcal{E}$, such that $m \stackrel{\varepsilon}{\to} n \stackrel{\gamma}{\to} c$ is a meaningful (i.e., observable) completion. We call this $\gamma$ a \emph{\ceo{}}.

To give a concrete example for the difference between $\mathcal{E}$, a generator $S$, and \eeo{s}, consider the modelling language SysML, which is used in system engineering~\cite{SysML2019}. The set $\mathcal{E}$ would then include very fine-grained elementary \eo{s} from $S$, for example, adding a port to a component or adding a connector between two ports, but also very large \eo{s}, for example for setting up the entire system architecture of a manufacturing execution system. An \eeo{} could be given, for example, by ``adding an interface'', that is, adding source port, target port, and connector in one step. After a port has been added to a component, a \ceo{} could be the completion of adding a target port and connecting the ports via a connector.


 



\todo[inline]{\begin{theorem}
Suppose the size of the connected components is constant and the 
and the number of pattern occurrences per subgraph is independent of the pattern. 
Then serializing graphs and sampling most probable edges is equivalent to transaction based frequent subgraph mining.
\end{theorem}

\begin{proof}

Step 1: The probability of edge extensions is equivalent to the maximum independent set frequency measure, which obeys the downward closure property. \par
Step 2: $f_{MIS} > f$, and therefore a probable edge extension will find frequent patterns. 
\todo[inline]{WIP}
\end{proof}}

\subsection{Automatic Inference of Edit Operations}
Since simple operations on the abstract syntax graph of a software model are very fine-grained and may lead to inconsistent models, and since structural model differences are too specific, there has been a long history of attempting to infer \emph{relevant} \eo{s} (and model transformations in general) from available information (e.g., the \mmodel{}, given examples, or a model history).

\begin{definition}[Automatic Inference of Edit Operations]
Given a \mmodel{} $\mathcal{TM}$ with models $\mathcal{M}$, \emph{automatic inference of edit operations} is a computable function 
$I\colon 2^{\mathcal{M}} \to 2^{\mathcal{E}}$ that, given a set of models (and their \mmodel{}) $\mathcal{M}' \subset \mathcal{M}$, computes a set of \eo{s} $\mathcal{E}' \subset \mathcal{E}$.
\end{definition}

There are \emph{supervised} and \emph{unsupervised} approaches to the inference of \eo{s}.
One branch of supervised approaches are demonstration approaches, were a tool user presents the transformation steps to an operation recorder \cite{Yu2011,Langer2009}. Typically, these approaches require some manual post-processing, for example, edit conditions have to be refined manually \cite{Langer2009}.
Another branch of supervised approaches include by-example approaches, with which the tool user specifies a set of examples (and sometimes also counter examples) and the tool automatically infers the model transformations \cite{Jalali2012, Kehrer2017}. These approaches have been motivated by the seminal work of Varró \cite{Varro2006}, who proposed by-example learning for exogenous model transformations.
Furthermore, heuristic approaches \cite{Mokaddem2018} applying search-based techniques to finding a set of refactorings have been proposed. 

Unsupervised approaches include generative approaches, deriving model transformations from the \mmodel{}, and mining approaches.
Generative approaches have been proposed in the area of model transformation testing \cite{Brottier06mmfuzz}. Also, more recently, generative fuzzing approaches based on language models have been proposed that try to generate models with similar properties to real-world models~\cite{Shrestha2021SLGPT}. \Eo{s} are only indirectly addressed within these generative approaches. It has been shown that a complete set of consistency preserving \eo{s} can also be derived from the \mmodel{} \cite{Kehrer2016, mazanek2009generating}. These operations capture static constraints that are already present in the \mmodel{} and are typically very simple operations.
More recent MTBE approaches also use neural networks to learn exogenous model transformations, for example, Burgueño et al.~\cite{Burgueno2019, Burgueno201922} investigate learning exogenous model transformations from examples using long-short-term memory neural networks. However, these approaches need concrete input-output model pairs and have not been evaluated for endogenous model transformations.

Recently, mining model transformations from the modelling history using graph mining approaches has been proposed~\cite{Tinnes2021}.
\todo[inline]{Furthermore, in model optimization, it has been proposed to apply meta-heuristics... Also comparison of meta-heuristic and RL has been done... I think this is not close enough to mention here.
RL based Model Transformation (for optimization): https://ieeexplore.ieee.org/document/9592463/
}
An advantage of mining approaches over by-example approaches is that they do not require to handcraft examples. A disadvantage is their computational complexity and that negative application conditions and multi-object patterns (e.g., creating a variable count of elements in a model transformation) can not inferred.  Anyway, a post-processing (e.g., using by-example approaches) is conceivable, and it this sense, by-example approaches and mining approaches are orthogonal.

\subsection{Software Model Completion}
\emph{(Software) model completion} is the task of further evolving a software model based on a given (partial) model. More formally:
\begin{definition}[Model Completion]
    Given a set of model transformations $\mathcal{T}$, \emph{model completion} is a computable function $C\colon \mathcal{T} \to \mathcal{T}$ that, given a model transformation $m \stackrel{\varepsilon}{\to} n$ from a source model $m$ to a (partial) target model $n$, computes a model transformation $C(m \stackrel{\varepsilon}{\to} n) = n \stackrel{\gamma}{\to} c$.  
\end{definition}

 The problem of model completion is to find a meaningful completion $m \stackrel{\varepsilon}{\to} n \stackrel{\gamma}{\to} c$. We call the \eo{} $\gamma$ a \ceo{}.
 To implement model completion, some authors proposed to automatically complete a partial model to a model that conforms to the meta-model and possibly other constraints(e.g., via rule based approaches or constraint solving~\cite{ohrndorf2021history,Taentzer2017,sen10completion,steimann2013}). However, these approaches are only able to complete a partial model to a model that conforms to the \mmodel{} or satisfies additional explicitly given constraints.
Other works take existing model repositories or pattern databases into account and employ model clone detection to discover similar (parts of) existing models to make completion recommendations~\cite{Stephan19SimulinkCompletion, Chowdhury2014}. Fine-tuned large language models would come in here very handy, because they encode the software model history in their parameters and can provide context-specific completions. The would not require hand-crafted pattern databases or expensive clone detection in the model repository and are able to generalize to unseen context information. Nevertheless, we are not aware of existing work investigating language models for software model completion.
An overview of model completion approaches is given in the secondary study by Almonte et al.~\cite{almonte22}.

\subsection{Language Models}\label{sec:lm}
\begin{definition}[Language Model]
    A \emph{language model} is a conditional probability distribution $\mathbb{P}(\omega|c)$ for a (sequence of) token(s) $\omega$, given a sequence of context tokens $c$. 
\end{definition}

The probability distribution is typically derived from a \emph{corpus} of documents, containing (some of) the tokens.
This has been done by means of statistical methods in the past. For example, Jelinek and Mercer \cite{jelinek1980ngram} derive the probabilities of so-called n-grams (i.e., sequences of n tokens) from a corpus. 
Motivated by the representation of concepts and language in the human brain, the idea of \emph{distributed representations}~\cite{hinton1986distrepr} has been developed. The idea is that there are efficient representations of (sequences of) words that avoid handling sparse n-gram tables and can carry linguistic and semantic information. Based on this idea, Bengio et al. \cite{bengio00neural} proposed using neural network architectures to learn the probability distribution $\mathbb{P}(\omega|c)$. With the success of transformer architecture \cite{vaswani2017attention}, these models have become quite popular now and are used in plenty of domains including software engineering~\cite{samoaa2022systematic,zhao2021natural,Xu2022LLMCode}. Today neural language models with billions of parameters are \emph{pre-trained} on large corpora of natural language text as well as source code. These pre-trained models can then be \emph{fine-tuned} for specific data sets and applications. It is also worth noting that language models are \emph{generative models}, that is, a language model can be used to generate new sequences of tokens whose probabilistic distribution should approximate the distribution of tokens in the corpora used in training.

\subsection{Further Related Work}

An area of research related to inference of \eo{s} is automatic ``by-example'' program synthesis~\cite{gulwani2017program} such as Flash Meta~\cite{Oleksandr2015}. From the perspective of this kind of research, in the approach presented in this work, we try to learn a grammar. We haven't evaluated language models yet for identifying functional relationships between attributes in software models. Instead, our approach is able to detect patterns in the data and does not need any example pairs to derive these patterns. Therefore, the program synthesis approaches are more similar to the classical model transformation by-example approaches.
For the future, a combination of symbolic inductive programming approaches with neural approaches as the one presented here might be an interesting field of research.

Some researchers investigate the extraction of knowledge graphs from language models~\cite{Swamy2021LM2KM, Garg2022CanLM}. Our idea of extracting \eeo{s} from a fine-tuned language model is similar in that we also try to extract ``learned knowledge'' from a language model. Anyway, our application domain is a different one. 
In cheminformatics and bioinformatics, language models have been investigated for molecular generation, that is, to the generation of (undirected) graphs~\cite{zhumagambetov2021transmol, bagal2021molgpt}. 







Regarding the application of natural language processing and language models, in particular there has been some research activities in the model-driven engineering community:
Burgue{\~n}o et al. propose an NLP-based architecture for the auto-completion of partial domain models. They do not employ language models in their approach and instead use other natural language processing approaches (i.e., word embedding similarity) to recommend words, which are then transformed into model elements in a post-processing~\cite{burgueno2021nlp}.
Weyssow et al. use a long-short-term memory neural network architecture to recommend \mmodel{} concepts but they do not generate entire model completions~\cite{weyssow2022recommending}. 


Furthermore, it needs to be mentioned that for source code, the use of language models for code completion is outperforming other approaches and capable of generating code from natural language input~\cite{chen2021codex, ciniselli22}. Sobania et al.~\cite{Sobania2021GPvsGPT} compare large language model-based GitHub CoPilot to genetic programming based program synthesis. Wan et al.~\cite{wan22} study if abstract syntax tree representations can be retrieved from language models trained on code.
\section{Approach}\label{sec:approach}
\subsection{Motivation}\label{sec:motivation}
As discussed in \secref{sec:eos}, we are interested in identifying \eo{s} that are \emph{meaningful}. Of course, which \eo{} can considered to be meaningful is highly task dependent. As we are interested in the evolution of software models, meaningful \eo{s} could be the ones helping to understand the model evolution. There is evidence that \eo{s} that compress the modelling history the most are meaningful to domain experts~\cite{Tinnes2021}. 

In this spirit, we are interested in identifying the most compressing (pattern) subgraphs in observed \scg{s} (i.e., the ones derived from to successive software model revisions).
Approaches to identify frequent or compressing subgraphs \cite{Yan2002, cook2006mining, Ketkar2005} of a database of graphs often grow the patterns (also called motifs) edge-wise. That is, they start with a single node and extend it edge by edge. Given a context graph $G$, two possible extensions edges $e$ and $e'$, and some metric $M$ (e.g., frequency or compression), the search then favours the extension $e$ with better metric $M(e \cap g) > M(e' \cap g)$ (e.g., in the form of a beam search as in Subdue \cite{Ketkar2005}).  
For a frequency and a compression measure, this condition can then be reformulated to yield a probabilistic formulation:
\begin{alignat}{2}
    &\qquad M(e \cap G) 
    \,&>&\quad M(e' \cap G)\\
    \Leftrightarrow 
    &\qquad \mathbb{P}(e \cap G) 
    &>&\quad \mathbb{P}(e' \cap G)\label{eq:pq2}\\
    \Leftrightarrow 
    &\qquad \frac{\mathbb{P}(e \cap G)}{\mathbb{P}(G)} 
    &>&\quad \frac{\mathbb{P}(e' \cap G)}{\mathbb{P}(G)} \\ 
    \Leftrightarrow 
    &\qquad \mathbb{P}(e \mid G) 
    &>&\quad \mathbb{P}(e' \mid G)\label{eq:pq4}
\end{alignat}
In (\ref{eq:pq2}), we have normalized over the whole dataset to yield a probabilistic formulation and in (\ref{eq:pq4}), we apply the definition of the conditional probability. In this formulation, the extension criteria reminds a lot on the language models from \secref{sec:lm}, with the major difference that language models are probability distributions on sequences of tokens, while the formulation above is for a probability distribution on sets of graphs. 
This suggests that language models can be used to generate serializations of patterns in \scg{s} in an edge-wise fashion.

\subsection{Concrete Approach}
The motivation of \secref{sec:motivation} suggests the following high-level procedure to employ language models for the mining of \eeo{s}: (1)~We serialize \scg{s} of (successive) pairs of software models edge-wise. (2)~We then generate pairs of partial \scg{s} and their completion to the full \scg{} serialization. (3)~These pairs are then used to fine-tune a pre-trained language model. (4)~The fine-tuned language model can then be used to generate serializations of \scg{s}. When a context is already given, we are in a completion setting. (5)~In a last step, we rank the generated serializations from the previous step.

We can divide the procedure into two phases: \trainingPhase{} (Step 1, 2, and 3) and the \generationPhase{} (i.e., Step 4 and 5). In what follows we will describe the steps in more detail.

\subsubsection{\TrainingPhase{}}
The \emph{input} to the \trainingPhase{} is a set of \scg{s} computed from pair-wise (successive) differences of software models in a model repository.
The \emph{output} of the \trainingPhase{} is a fine-tuned language model. 


\textit{Step~(1) -- Serialize the graph components: }
We have to serialize the graph as an edge list. One reason for this is that we want to sample the \scg{s} edge-wise, as suggested in \secref{sec:motivation}. Common formats such as the GraphML\footnote{\url{http://graphml.graphdrawing.org/}} are not suitable, because they start with a list of vertices before they list the edges. It is intuitive from a language model perspective why this is not a suitable format: Suppose our initial \scg{s} are the results of the application of more than one \eeo{}. In this case, the initial \scg{s} will be larger than the \scg{s} of the \eeo{s} that we want to discover. Anyway, the language model will learn to generate serializations statistically similar to the input. This implies that the number of nodes would probably always be too large for a serialization of an \eeo{}. 
We therefore use a simple graph serialization called \edgel{} for directed labeled graphs:
The serialization of every edge has the format
\begin{lstlisting}[frame=single, basicstyle=\sffamily\scriptsize, breakindent=0pt]
e <src_id> <tgt_id> <edge_label> <src_label> <tgt_label>
\end{lstlisting}
where \lstinline[basicstyle=\sffamily\small]{<src_id>} and \lstinline[basicstyle=\sffamily\small]{<tgt_id>} are identifiers for the source and target vertices of the edge, respectively.
The serialization of a graph starts with a header line~in the format 
\begin{lstlisting}[frame=single, basicstyle=\sffamily\scriptsize, breakindent=0pt]
t # <graph_id>
\end{lstlisting}
and then all edges of the graph serialized line~by line. An example is given in Listing \ref{list:excerpt}.


\begin{lstlisting}[frame=single, caption={An example \SCG{} in the \edgel{} format.}, basicstyle=\sffamily\scriptsize, breakindent=0pt, captionpos=b, label={list:excerpt}]
t # 1
e 0 1 Add_port Add_Component Add_Port
e 0 2 Add_requirement Add_Component Add_Requirement
\end{lstlisting}

Another degree of freedom that arises in the serialization step is the order of the edges and the identifiers of the vertices in this serialization. After the order of the edges is chosen, we can enumerate the vertices (e.g., increasing integers in the order they appear when following the edges). This still leaves us with $|E|!$ choices (up to automorphism), which can quickly become impractically large. 
There are \emph{canonical graph serializations}~\cite{McKay2014iso}, but they do have exponential worst-case time complexity and do not help for our task, because these serializations do not ensure that pattern subgraphs appear with their own canonical ordering\footnote{Otherwise graph isomorphism problem and subgraph isomorphism problem would be in the same complexity class, which is not thought to be the case.}. Therefore, choosing \emph{some} of the possible $|E|!$ edge orderings will probably lead to better results. In this work we chose one ordering (obtained through a depth-first search) that worked well for our data, and we leave the investigation of the influence of further edge orderings for future research. \todo[inline]{The main reason for this limitation is costs! we used GPT-3 language model series and also the davinci model in our experiments, which is quite expensive. Now that we know---as one result of this work---that davinci does not perform significantly better than cheaper models, it would be interesting to fix a cheaper language model like ada and blow more data, i.e., more serializations through it.}

\textit{Step~(2) -- Randomly split serializations in prompt and completion:}
As input for the fine-tuning of language model, we provide a set of context and completion pairs. From every \edgel{} serialization from the previous step, we generate three training samples: We compute three cut points. One cut point is randomly chosen in the first 10\% of the edges, the second randomly between the first and the last 10\% of the edges, and the third cut point randomly among the last 10\% of the edges. For every cut point, we then obtain one training sample by taking the edges before the cut point as context and the edges after the cut point as the completion.

\textit{Step~(3) -- Fine-tune a language model:}
In this step, we use the dataset obtained above to fine-tune a pre-trained language model. Specifically, we use so-called autoregressive language models (the GPT-3 model family), although other types of language models are also conceivable. In autoregressive language models, only the probabilities for the next token in a sequence is predicted based on the context (i.e., the previous tokens). This way, we obtain a language model that is fine-tuned to the \scg{} serializations computed for a specific model repository.

\subsubsection{\GenerationPhase{}}
The fine-tuned language model can now be applied to generate \eeo{s}, on the one hand, and \ceo{s}, on the other.

\textit{Step~(4) -- Generate \scg{s} for \eeo{s} and \ceo{s}:}
The main difference between the generation of \eeo{} candidates and \ceo{} candidates is the context. For the generation of \ceo{} candidates, we use the \scg{} serialization of an observed evolution $m_1 \stackrel{\varepsilon}{\to} m_{2}$ as the context in the generation of candidates. For the generation of \eeo{s}, we use an ``empty edge'' context (i.e., the token  \lstinline[basicstyle=\sffamily\small]{e} to be more precise).

The candidate generation works as follows (see pseudo code in Listing \ref{lst:candidateGeneration}):
The algorithm takes a set of \emph{incomplete} \eo{} candidates (in the form of serialized \scg{s}) and uses the fine-tuned language model to sample new edge candidates and appends them to the incomplete \eo{} candidate (Line~12). 
The sampling generates all possible extensions above a certain probability threshold.
Since we cannot guarantee that the extensions lead to a correct \edgel{} serialization, we check the syntactical correctness and reject incorrect extensions (Line~13). Furthermore, even syntactically valid extensions could be invalid according to the \mmodel{} and have to be rejected likewise (Line~14). After that, the corresponding \scg{} represents a valid \eo{} by definition. Based on a graph isomorphism test, we then filter out duplicates (Line~15). Although graph isomorphism is theoretically expensive from a computational perspective, in our setting, it is acceptable since we have only a few medium size graphs, and employ Weisfeiler-Lehman hashes~\cite{huang2021short} to speed up the comparison.  We add complete candidates to the output list (Line~19) and repeat this process until all candidates are complete (Line~9). Whether a candidate is complete is checked using several conditions such as the total probability of the candidate, a drop in the probability of a generated edge, or a generated stop token.

\begin{figure}[bt]
			\centering
			\includegraphics[width=\columnwidth]{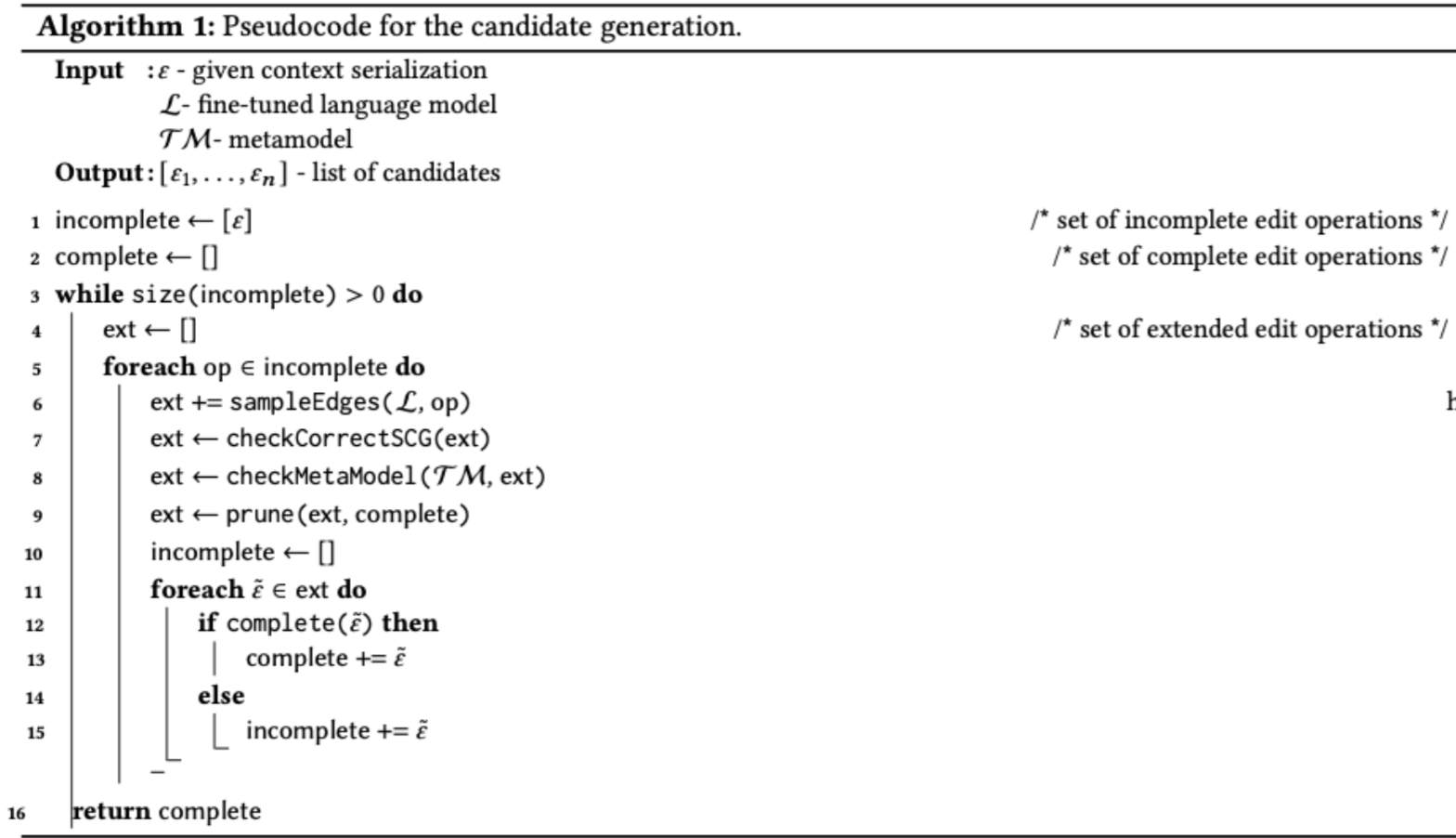}
			\caption{Pseudocode for the candidate generation.}
            \label{lst:candidateGeneration}
\end{figure}

\textit{Step~(5) -- Ranking of generated candidates:}
In the last step, we rank the generated candidates from Step 4. Several ranking metrics are conceivable, for example, the probability for a candidate given by the token probabilities of the language model, the compression metric as used in the work of Tinnes et al. \cite{Tinnes2021}, or also scaled variants of the token probabilities such as scaling with the number of edges or nodes. We will evaluate different ranking metrics.
\section{Evaluation}\label{sec:evaluation}
In this section we evaluate to what extent language models, as exemplified by our approach (\secref{sec:approach}), can help to derive \eeo{s} and \ceo{s} from the software model history.

\subsection{Research Questions}
To better understand the merits of language models for \eo{} and \ceo{} mining, we set out to answer the following research questions:


\begin{tcolorbox}[arc=0pt,outer arc=0pt, boxrule=0pt, breakable,  sharpish corners,enhanced, drop lifted shadow]
\textbf{RQ 1:} \textit{Using \approach{}, are language models capable of generating correct \scg{} graph serializations?}
\end{tcolorbox}
\vspace{-0.5em}
Since a language model is not aware of the meta-model and definition of a graph \emph{per se}, the generated \eeo{s} or \ceo{s} might not be correct \scg{} serializations. That is, they might be invalid according to the \mmodel{} (e.g., invalid combination of edge, source, and target node labels) or could even be invalid directed labeled graph serializations (i.e., not adhere to the \edgel{} format). We are particularly interested how this depends on the properties of the dataset and the properties of the language model used for the approach.

\begin{tcolorbox}[arc=0pt,outer arc=0pt, boxrule=0pt, breakable,  sharpish corners,enhanced, drop lifted shadow]
\textbf{RQ 2:} \textit{Using \approach{}, are language models capable of providing \ceo{s} for software models?}
\end{tcolorbox}
\vspace{-0.5em}
The approach uses a language model that is trained to complete the \scg{} serializations. It optimizes the \emph{token probability}, given a context. This does not ensure \emph{per se} that the provided completions represent \scg{s} that are isomorphic to the \emph{correct} \scg{s}. Therefore, we are interested in to which extent the completed and the original \scg{} coincide and how this depends on the properties of the dataset and of the language model.  

\begin{tcolorbox}[arc=0pt,outer arc=0pt, boxrule=0pt, breakable,  sharpish corners,enhanced, drop lifted shadow]
\textbf{RQ 3:} \textit{Using \approach{}, can \eo{s} be reconstructed from the language model?}
\end{tcolorbox}
\vspace{-0.5em}
The main idea of \approach{} is that the language model leverages patterns in the training data while generating text. Therefore, it should be possible to read out the patterns from the language model. The approach presented here is just one idea how the patterns can be retrieved from the language model, and we have to evaluate this approach empirically. We further evaluate how the generation of \eeo{} candidates depends on the properties of the dataset and the properties of the language model used for the approach.

\begin{tcolorbox}[arc=0pt,outer arc=0pt, boxrule=0pt, breakable,  sharpish corners,enhanced, drop lifted shadow]
\textbf{RQ 4:} \textit{Which ranking metric for the \eo{} candidates performs best?}
\end{tcolorbox}
\vspace{-0.5em}
There are several possibilities to rank the list of generated \eo{} candidates, including language model probability, the probability scaled by the factorial of the number of edges, or scaled by the number of edges, or a more computationally expensive compression-like metric as used in a previous approach~\cite{Tinnes2021}. The idea behind the scaled metrics is that -- as discussed already above -- there are several possible \scg{} serialization (up to  $|E|!$). At least, the probability will inevitably decrease with the number of edges and we have to account for this to avoid favouring smaller \eo{} candidates.

\subsection{Dataset}\label{sec:dataset}
As we will discuss in \secref{sec:threats}, the goal of this work is to study the merits of language models for \eeo{} and \ceo{} mining with a high internal validity. To answer the research questions, we simulated the evolution of a software model similar to previous work~\cite{Tinnes2021}. This gives us control over the \eo{s} that have been applied to yield the model history. For this simulation, we used a \mmodel{} from Tinnes et al. \cite{Tinnes2021} that resembles a simple component model (as used in modelling system architecture) with {\sf components}, {\sf implementations}, {\sf ports}, {\sf connectors}, and {\sf requirements}. We then randomly apply \eo{s}. Specifically, we applied three different kinds of edit operations (i.e., adding a component, adding an interface, and adding a new package including a component). We controlled for the number of edit operations that are applied per model revision (i.e., 11, 31, 51, 81) and the number of model revisions in one dataset (i.e., 10 or 20). We furthermore randomly applied perturbations, that is, with a certain probability (i.e., 0\%, 50\%, 100\%), we slightly modified the edit operation by a successive application of an additional \eo{} that overlaps with the original \eo{}. This way, we obtain \countRepo{} simulations of a software model evolution (see \figref{fig:dataset}) and we know which \eo{s} had been applied between two model revisions of the repositories.

\begin{figure}[bt]
			\centering
			\includegraphics[width=\columnwidth]{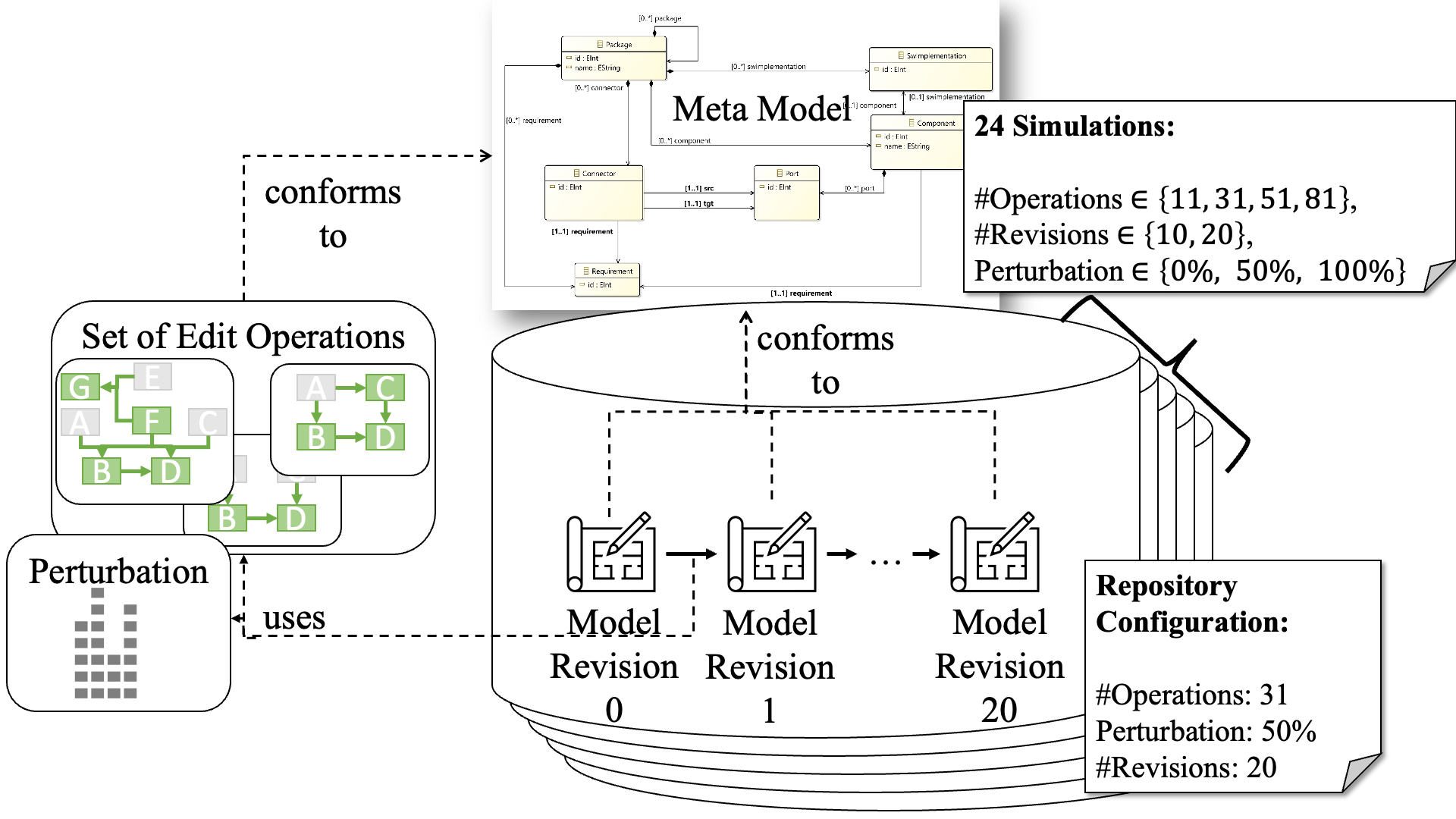}
			\caption{Simulation of the model repositories.}
			\label{fig:dataset}
\vspace{-2em}
\end{figure}

\subsection{Operationalization}
    For every simulated repository in the dataset from \secref{sec:dataset}, we applied the steps of the \trainingPhase{}, described in \secref{sec:approach}. We controlled for the number of epochs (i.e., 4 and 6) and the base language model used for the fine-tuning (i.e., \ada{}, \curie{}, and \davinci{} from the GPT-3 family of language models). Since fine-tuning the \davinci{} model is quite expensive (i.e., 3 Cents per thousand tokens at the time of writing), we fine-tuned this model only for the model repositories where perturbation probability equals 100\% (the ones which are typically the harder ones). We therefore have to report separately if we also include the \davinci{} model in the analysis. We split the datasets (for each of the repositories from \secref{sec:dataset}) in a training and a test set (90\% of the samples in the train set and 10\% in the test set), so that we can report on the performance of the (textual) completion task the models were trained on.

    The total cost for the training via the OpenAI API was \totalCost{}.
    
	\textbf{Experiment 1:} To answer the first research question, we use the fine-tuned language models and apply the \eeo{} candidate generation from \secref{sec:approach}. During the edge extension in the \generationPhase{}, we checked whether the completion actually corresponds to a correct graph serialization. We also checked whether the edge extension corresponds to a correct extension according to the meta-model, that is, whether the generated graph actually corresponds to a \scg{}. We then report how many of the extensions during the generation were incorrect and how the number of incorrectly generated serializations depends on the base language model, the number of training epochs, the perturbation parameter, and the number of training tokens in the dataset.
	
	\textbf{Experiment 2:} We fine-tune the language models based on a model completion task. Therefore, the procedure from \secref{sec:approach} for deriving fine-tuning data from a model repository together with the training test split described above yield data that can be used to evaluate whether language models can be used for model completion, answering the second research question. 
    We investigate the \scg{} completion from two perspectives:
    First, we investigate the average token accuracy on the test set during the fine-tuning of the language model. The  \emph{average token accuracy} gives us the relative number of correctly retrieved tokens. The metric is not aware of any specifics of the dataset. For example, even a single wrong token in a serialization can produce a syntactically wrong serialization while the token accuracy for this can still be high. We therefore also analyze the \ceo{} candidates from a graph matching perspective. Since generating all completion candidates for all test samples of all fined-tuned language models will be quite expensive, we select \numCompletionModels{} fine-tuned language models and perform the analysis for them. From the set of generated \ceo{} candidates, we especially look at two candidates: the one providing ``the best completion'' and the ``top-ranked completion'' using the edge-scaled ranking metric (see Experiment 4). In some sense, these two selected candidates give us an upper and a lower bound for the performance on the software model completion task. To make completions comparable, we assign a numeric value to them. This is possible, since we define completion candidates isomorphic to the correct graph to be better than completion candidates that are too large (i.e., the ground truth is a subgraph of the completion candidate), which themselves are better than completion candidates that are too small, which again are better than incorrect completion candidates (i.e., some edges missing and some additional edges).
    

	\textbf{Experiment 3:} To answer RQ 3, we evaluate the generation of \eeo{} candidates from the fine-tuned language models. To this end, we apply the approach from \secref{sec:approach} to the synthetic model repositories. Since we know the \eo{s} that have been applied, we can directly look for the applied \eo{s} in the list of \eeo{} recommendations. We count the number of correct \eo{} candidates that have been generated. There are three correct \eo{s} in total, and each of them has been applied in every software model to our synthetic dataset (see \figref{fig:dataset}). Additionally, we investigate how the number of correctly retrieved \eo{s} depends on repository as well as language model parameters. We will also investigate the costs for the generation of \eeo{} candidates.


    
    \textbf{Experiment 4:} As in experiment 3, we apply the approach to the synthetic model repositories. We compute the different ranking metrics (i.e., language model probability, the probability scaled by the factorial of the number of edges, probability scaled by the number of edges, and compression-like metric) for all generated \eeo{} candidates. We then compare the different rankings based on the given ground truth, that is, the rank of the known applied \eo{}.
    To compare the different ranking metrics, we use the \emph{mean average precision at k} (MAP@k), which is commonly used in the evaluation of recommender systems~\citep{schroder2011setting}:
    \begin{equation*}
    \operatorname{MAP@k} := \frac{1}{|D|} \sum_{\text{D}} \operatorname{AP@k} \, \text{,}
    \end{equation*}
    where $D$ is the family of all datasets (one dataset represents one repository) and AP@k is defined by 
    \begin{equation*}
    \operatorname{AP@k} := \frac{\sum_{i=1}^{k} \operatorname{P}(i) \cdot \operatorname{rel}(i)}{{|\text{all correct \scg{s}}|}} \, \text{, } 
    \end{equation*}
    where P($i$) is the precision at $i$, and rel($i$) indicates if the candidate at rank $i$ is relevant. 

\subsection{Results}
\textbf{Experiment 1:}
For 51.8\% of the simulated repositories, the generation procedure produced exclusively valid graphs, for 48.2\% it produced also only correct \scg{} (i.e., also correct with respect to the \mmodel{}). On average, $2.26$ invalid graphs are generated in the \generationPhase{}, with a minimum of $0$ and a maximum of $30$. A constraint in the given \edgel{} serialization is that a node with a given id has to appear always with the same label (i.e., the node labels are redundantly encoded in \edgel{}). The only type of violation against a valid \edgel{} encoding we have encountered in this experiment was that this correspondence of node id and node label has been violated. 
A manual inspection of the data for the model repositories with a large amount of invalid graphs ($>5$) reveals that these are the ``smaller'' datasets with mostly a high perturbation. For example, the repository with only 10 revisions, 11 applied \eo{s}, and a perturbation of 100\% is the one with the maximum of $30$ invalid graphs.

In \tblref{tbl:cor_faulty_graphs} we report the correlation coefficients (Spearman\footnote{We use Spearman correlation instead of Pearson correlation, since we can not assume that there are linear dependencies.}) between invalid graphs and invalid \scg{s}  and the base language model (BM), the number of training epochs, the perturbation probability (P), and the number of training tokens of the dataset (T). For the correlation with the base language model, we sort the base models according to their size (i.e., ada:~0, curie:~1, davinci:~2).

\begin{table}[bt]
\caption{Correlation (Spearman) between invalid graphs generation with other parameters. (**: $p < .001$, *:~$p < .01$)}
\centering
\resizebox{\linewidth}{!}{%
\begin{tabular}{rSSSS}
  \toprule
 & {BM} & {Epochs} & {P} & {T}\\ 
  \midrule
  \#Invalid &  -0.08  & -0.10  &  0.35\,** & -0.33*  \\ 
  \#Invalid \mmodel{} & -0.22  & -0.05  & 0.31*  & -0.22   \\ 
  \#Invalid (+ davinci) & -0.28  & -0.11  & \text{--} & -0.56** \\ 
  \#Invalid \mmodel{} (+ davinci) & -0.24  & -0.13  & \text{--} & -0.56**    \\ 
   \bottomrule
\end{tabular}}
\label{tbl:cor_faulty_graphs}
\vspace{-1em}
\end{table}

We observe significant positive Spearman correlation between the number of invalid generated graph serializations and the perturbation parameter. Furthermore, there is a significant negative Spearman correlation between the number of invalid generated graph serializations and the number of tokens in the training set. 

\textbf{Experiment 2:}
At the token level, we find an average token accuracy of \averageTokenAccuracy{}, with a minimum of \minTokenAccuracy{}, and a maximum of \maxTokenAccuracy{} on our test data set.

Only \averageNumberCompletions{} completion candidates are generated, on average. For a large number of samples, only one \ceo{} candidate has been generated. In all of theses cases, the only candidate has also been the correct one.
On average, in \averageCorrectCompletion{\%} of the samples, the correct completion is among the candidates, with an average rank of \averageRankCompletion{} (w.r.t. the edge-scaled ranking).

In \tblref{tbl:cor_completion} we report on the Spearman correlation coefficients of the score of the completion candidates with the number of omitted edges (i.e., the ones that have to be computed), the total number of edges of the ground truth \scg{}, and the number of completions that have been generated.

\begin{table}[bt]
\caption{Correlation coefficients (Spearman) of the completion candidate score with the number of omitted edges in the sample, the number of total edges of the correct \scg{}, and the number of completion candidates that have been generated. (*: $p < .001$)}
\centering
\resizebox{\columnwidth}{!}{%
\begin{tabular}{rSSS}
  \toprule
 & {\#Omitted Edges} & {\#Total Edges} & {\#Completions}\\ 
  \midrule
  Score (best rank) & -0.69* & -0.38* &  -0.75*\\ 
  Score (best candidate) & -0.29* & -0.33* & -0.21*\\ 
   \bottomrule
\end{tabular}}
\label{tbl:cor_completion}
\end{table}

We see significant negative correlations between the score (of the best generated candidate and the best ranked candidate) and the number of edges that have to be completed, the number of edges from the full original \scg{}, and the number of completions candidates that have been generated. 

\textbf{Experiment 3:}
On average, out of the three applied \eo{s}, we could retrieve \averageCorrectAda{} for the \ada{} model, \averageCorrectCurie{} for the \curie{} model, and \averageCorrectDavinci{} for the \davinci{} model. The \davinci{} model has only been trained to the datasets that appeared to be difficult for \ada{} and \curie{} (i.e., perturbation probability of 100\%). 

In \tblref{tbl:cor_retrieved}, we list the Spearman correlations of the correctly retrieved \eo{s} with the base model (BM), the number of training epochs, the perturbation probability (P), and the number of applied \eo{s} between two model revisions~(E).

\begin{table}[bt]
\vspace{-0.75em}
\caption{Correlations between the correctly retrieved \eo{s} and repository as well as language model parameters. (*: $p < .001$)}
\centering
\resizebox{\columnwidth}{!}{%
\begin{tabular}{rSSSS}
  \toprule
 & {BM} & {Epochs} & {P} & {E}\\ 
  \midrule
   \#Correct &  -0.33*  & -0.03  &  -0.80* & 0.10  \\ 
   \#Correct (+ davinci) & -0.22  & 0.08  & \text{--} & 0.28    \\ 
   \bottomrule
\end{tabular}}
\label{tbl:cor_retrieved}
\vspace{-1em}
\end{table}

The average generation cost for \ada{} was 0.45~Cent, for \curie{} 3.68 Cent, and for \davinci{} 51.84~Cent.

\textbf{Experiment 4:}
We list the MAP scores for the 4 different ranking metrics in \tblref{tbl:map_scores}.
Furthermore, we compare the average precisions obtained through the different ranking. We can observe a high significant ($p < .001$) Spearman correlations coefficients~$>0.65$ among all of them, the largest one between compression metric and node factorial scaled probability metric ($0.90$).

\begin{table}[b]
\caption{MAP for the different ranking metrics. Grey background indicates the best MAP among all metrics.}
\centering
\resizebox{\columnwidth}{!}{%
\begin{tabular}{lSSSS}
  \toprule
 & {Compression} & {Factorial} & {Probability} & {Edges-Scaled}\\ 
  \midrule
  MAP@3 & 0.32 & 0.20 & 0.13 & \cellcolor[gray]{0.9} 0.33 \\ 
  MAP@5 & \cellcolor[gray]{0.9}0.38 & 0.29 & 0.17 & 0.36 \\ 
  MAP@10 & \cellcolor[gray]{0.9}0.41 & 0.32 & 0.24 & 0.40 \\ 
  MAP@$\infty$ & \cellcolor[gray]{0.9}0.42 & 0.33 & 0.25 & 0.41 \\ 
   \bottomrule
\end{tabular}}
\label{tbl:map_scores}
\vspace{-2em}
\end{table}

\subsection{Discussion}
\textbf{RQ 1:} From our results, we can conclude that with our current approach we generate mostly valid graph serializations. Indeed, even for a majority of the simulated repositories, we do not generate invalid graph serializations, at all, and, on average, \avgIncorrectGraphs{} invalid candidates per \generationPhase{}. If we consider correct \scg{s}, only \avgIncorrectScgs{} illegal \scg{s} (i.e., they do not conform to the \mmodel{} or the definition of \scg{s}) are generated, on average. The analysis also shows that the repositories for which we also get invalid graph serializations are the smaller ones (in the number of training tokens), with a high perturbation. The dependency on the size of the repository and the perturbation is also significant (as can be seen in \tblref{tbl:cor_faulty_graphs}). This suggests, that one should use larger repositories or even pre-train the language model with \scg{} serializations from other repositories. We also observe a small negative but insignificant correlation w.r.t. the size of the base language model that has been used for the fine-tuning. This suggests that larger base language models might perform better in the generation of \scg{}. Given that \davinci{} is 50 times as expensive as \ada{}, and this correlation is not significant, we can conclude that \ada{} is an acceptable choice for the base language model for the generation of \scg{s}. 

\begin{tcolorbox}[arc=0pt,outer arc=0pt, boxrule=0pt, breakable,  sharpish corners,enhanced, drop lifted shadow]
\vspace{-0.5em}
\textbf{Summary:} 
Overall, invalid graphs are generated only rarely. Smaller repositories are more likely to lead to invalid graphs than larger repositories. 
\end{tcolorbox}

\textbf{RQ 2:}
From the results, we can clearly conclude that, in a majority of the samples, the correct completions have been generated. Also, the number of \ceo{} candidates is typically low and the correct \ceo{} (if among the candidates) is typically top ranked. The larger the \scg{} and the more edges we omit for the completion, the worse the score of the completion candidates. Anyway, for the larger graphs, a subgraph of the correct completion was among the candidates in most cases.

\begin{tcolorbox}[arc=0pt,outer arc=0pt, boxrule=0pt, breakable,  sharpish corners,enhanced, drop lifted shadow]
\vspace{-0.5em}
\textbf{Summary:} 
A manageable amount of \ceo{} candidates are generated and in a majority of the cases the correct completion has been generated by the language model. 
Larger \scg{s} (in number of edges) and a larger number of omitted edges are more challenging than smaller ones, which is rather intuitive.
\end{tcolorbox}

\textbf{RQ 3:}
Overall, using the approach from \secref{sec:approach}, we were able to retrieve two and, in some cases, even all of the applied \eo{s}. The perturbation seems to be the major influencing factor that increases the difficulty.
Using larger language models (i.e., \davinci{}) did not improve the results. We can even see a decrease of the number of retrieved \eo{s} with increasing language model size.
The costs for using the language models are definitely acceptable (0.45 Cent to 51.84 Cent), especially given that the more expensive language models do not perform any better.

\begin{tcolorbox}[arc=0pt,outer arc=0pt, boxrule=0pt, breakable,  sharpish corners,enhanced, drop lifted shadow]
\vspace{-0.5em}
\textbf{Summary:} 
We were able to retrieve \eeo{s} from a fine-tuned language model. However, with increasing perturbation probability, the approach yields worse results. 
\end{tcolorbox}

\textbf{RQ 4:}
From the results of Experiment 4, we can see that, for $k>3$ (where $k$ is the number of considered generation candidates), the compression metric outperforms the other metrics. Since the compression metric is expensive to calculate, we also tried to recompute it from the probability given by the language model. Although we can observe high correlations among the average precisions, none of the three metrics yields the same ranking as the compression-based ranking. For practical purposes, the edges-scaled probability metric is most feasible, since it gives results close to the compression metric and does not require any expensive calculations.

\begin{tcolorbox}[arc=0pt,outer arc=0pt, boxrule=0pt, breakable,  sharpish corners,enhanced, drop lifted shadow]
\vspace{-0.5em}
\textbf{Summary:} 
We can confirm earlier results~\cite{Tinnes2021} that the compression metric seems to be a good metric to select \eeo{s}. There seems to be a -- yet unknown -- relationship between the probability given by the language model and the compression metric from earlier work.
\end{tcolorbox}


\textbf{Frequent subgraph mining vs. language models:}
In a recent study \cite{Tinnes2021}, frequent subgraph mining has been proposed for \eo{} mining.
First, this approach cannot be applied to model completion out-of-the-box, because it does not yield a conditional distribution that can be used for completion. 
Second, regarding \eo{} mining, frequent subgraph mining is challenging from a computational/memory complexity point of view. This sometimes requires omitting large \scg{s} in the mining. The reason for this is that subgraph matching requires to try many possible combinations that vertices from the subgraph can be matched to the super graph. 
Using language models, we circumvent this issue and train on sequential -- instead of graph-like -- data. The training resources (and therefore also cost) scale linearly with the number of tokens. A graph would be ``equivalent'' to all of its possible serializations (i.e., approximately~$|E|!$). Training a language model on all of theses serializations would then be infeasible. Anyway, using a language model as described in \secref{sec:approach}, we get control over the number of serializations and we have seen that considering only one serialization per graph yields promising results. 

\subsection{Threats To Validity}\label{sec:threats}
We have evaluated the proposed approach in a controlled experiment setting, this way, maximizing internal validity. However, further experimentation and evaluation is required before it can be considered for implementation in real-world software engineering projects. Our research is currently in the simulation phase. This stage of research is similar to the early stages of drug development in the pharmaceutical industry, where candidates for a treatment are first evaluated through simulations and small-scale experiments before moving on to larger studies. 
One of the main reasons to take such a staged approach is that the application of language models to large-scale industrial experiments requires training on large model repositories. Training on all possible simple change graphs and serializations is infeasible from a cost perspective. We therefore have to better understand large language models in the domain of model generation to use a suitable sampling in real-world applications. Furthermore, without cost-intensive legal considerations, real-world data can not be uploaded to the GPT-3 API and therefore we would have to rely on much smaller, non-state-of-the-art language models. 
So, we made the explicit decision to trade-off external validity for internal validity~\cite{siegmund2015views} before moving on to the next stage, that is, an application to real-world models. 
The purpose of the present study is to understand the merits of language models for \eeo{} and \ceo{} mining. 


With respect to internal validity, we have chosen those properties and parameters that intuitively have the highest impact on fine-tuning language models, although we cannot control for any arbitrary property of the modelling language or the model repository (because of combinatorial explosion).
Several design decisions (e.g., which parameters to fix and which to vary) have to be made before language models can be applied to the domain of software model completion (e.g., serialization strategy, graph encoding, choice of the base model, etc.). Other design decisions could have led to other conclusions and there is still room for improvement and optimization of language models for \eeo{} mining and model completion.
Furthermore, the base models (i.e., GPT-3) that were fine-tuned in this work are only available through an API. We wanted to conduct this study with state-of-the-art models, and at the time the experiments were conducted, GPT-3 was the state-of-the-art without significant competition. However, it would also be interesting to compare a larger set of language models on the above tasks, although this is well beyond the scope of the present study.
\section{Conclusion}
In this work, we have evaluated the principle feasibility of using language models (i.e., GPT-3) for extracting \eeo{s} from software model repositories and for model completion. A key advantage is that a language model has to be fine-tuned only once and can then be used for \eeo{} mining and model completion. We presented a formalization of both tasks that makes the synergy between them more obvious. 

To evaluate the use of language models for \eeo{} mining and model completion, we conducted a controlled experiment with a synthetic software model repository. 
We found promising results for both tasks (\eeo{} mining and model completion), even without elaborate optimizations. Our approach fine-tunes language models on a graph representation of model differences. It was able to correctly complete a majority of serialized difference graphs from a synthetic test dataset. Furthermore, some of the \eo{s} that have been applied to generate the synthetic difference graphs could be generated via the fine-tuned language model, although we were not able to always generate all applied \eo{s}. 

For future work, it is necessary to optimize the approach, for example, investigating the influence of adding more \scg{} serializations during the fine-tuning of the language model. Another way of optimization is to further fine-tune the approach in a supervised manner, on a set of correct and incorrect completions or a set of known \eeo{s}. Also, in a productive environment, user-feedback may be involved in the generations by wrapping the language model in a reinforcement learning approach.

Optimized versions of the approach have to be evaluated in more realistic, real-world settings.
Also, it would be interesting to investigate whether language models are able to learn attribute relationships (e.g., functional relationships between attributes) and learn to abstract (e.g., automatically discover ``is-a'' relationships).

\section{Data Availability}
We have provided data as well as Python scripts for our approach---to replicate the results of this paper---as a replication package. We furthermore provide the R scripts to replicate our statistical evaluation. The replication package will be made public (e.g., GitHub) in case of acceptance.


\def\bibfont{\footnotesize}
\bibliographystyle{plain}
\bibliography{bib}

\end{document}